\documentclass[reqno]{amsart}

\usepackage[normalem]{ulem}%pour barrer du texte

\usepackage{amsmath,amssymb,color}
\usepackage{amsfonts, amscd, epsfig, amsmath, amssymb,enumerate, bbm}
\usepackage{graphicx}
\usepackage{graphics}
\usepackage{color}
\usepackage{mathrsfs}
\usepackage{stmaryrd}
\usepackage{todonotes}
\usepackage{soul}
 \usepackage{relsize}

 \usepackage[usenames,dvipsnames]{pstricks}
 \usepackage{pstricks-add}
 \usepackage{epsfig}
 \usepackage{pst-grad} % For gradients
 \usepackage{pst-plot} % For axes
 \usepackage[space]{grffile} % For spaces in paths
 \usepackage{etoolbox} % For spaces in paths
\usepackage[margin=3cm]{geometry}
 \makeatletter % For spaces in paths
 \patchcmd\Gread@eps{\@inputcheck#1 }{\@inputcheck"#1"\relax}{}{}
 \makeatother

\newtheorem{theorem}{Theorem}[section]

\newtheorem{corollary}[theorem]{Corollary}

\usepackage{tikz}
\usetikzlibrary{backgrounds}
\usetikzlibrary{patterns,fadings}
\usetikzlibrary{arrows,decorations.pathmorphing}
\usetikzlibrary{decorations}
\usetikzlibrary{calc}
\usetikzlibrary{shapes.misc}

\usepackage{setspace}
\definecolor{light-gray}{gray}{0.95}
\usepackage{float}
\usepackage[colorlinks=true,linkcolor=blue,citecolor=magenta]{hyperref}
\def\centerarc[#1](#2)(#3:#4:#5){\draw[#1] ($(#2)+({#5*cos(#3)},{#5*sin(#3)})$) arc (#3:#4:#5);}

\renewcommand{\epsilon}{\varepsilon}

\newcommand{\R}{\mathbb R}

\newcommand{\Z}{\mathbb Z}
\newcommand{\N}{\mathbb N}

\newcommand{\T}{\mathbb T}
\newcommand{\E}{\mathbb E}

\newcommand{\eq}[2]{
\begin{equation}
\label{#1}#2
\end{equation}
}
\newcommand{\eqs}[1]{
\begin{equation*}
#1
\end{equation*}
}

\renewcommand{\tilde}{\widetilde}

\newcommand{\normm}[1]{{\left\vert\kern-0.1ex\left\vert\kern-0.1ex\left\vert\; #1 \; \right\vert\kern-0.1ex\right\vert\kern-0.1ex\right\vert}}    
\newcommand{\cro}[1]{\left[#1\right]}
\newcommand{\pa}[1]{\left(#1\right)}

\newcommand{\var}{\operatorname{Var}}

\renewcommand{\leq}{\leqslant}
\renewcommand{\geq}{\geqslant}

\renewcommand{\ge}{\geqslant}

\allowdisplaybreaks %%Pour eviter les grands espacements verticaux.

\title{Scaling relations for the CLG's critical exponents}

\author{Cl\'ement Erignoux, Assaf Shapira and Marielle Simon}
\begin{document}

\maketitle

\begin{abstract}
We consider, in any dimension, the constrained lattice gas introduced by physicists \cite{RPV}, which is an exclusion process on a $d$-dimensional lattice following the additional constraint that only particles with at least one occupied neighbour can jump.  In dimension $d\geq 2$, this model features self-organized criticality at some critical density of particles. Numerical simulations predict the existence of scaling exponents close to criticality, and several relations can be derived between these exponents. The goal of this article is to give a mathematical framework for these relations, which have been numerically established in a companion article \cite{ERSS24}.
\end{abstract}

\begin{center}
\textit{Dedicated to Claudio Landim for his 60th birthday}
\end{center}

\section{Introduction}
The \emph{Constrained Lattice Gas} (CLG) is an interacting particle system  introduced in
\cite{RPV}, describing the evolution of a $d$-dimensional lattice gas with exclusion rule, under the additional constraint that isolated particles are not allowed to move. This model has a dynamical critical point, separating a phase with quick absorption from a phase with sustained activity. With proper boundary condition, this model is an important example of \emph{self-organized criticality}. Since its introduction in \cite{RPV}, this model has suscitated a growing interest, both in the mathematics and physics community. The macroscopic behavior of the empirical particle density has been investigated in several recent papers (sometimes dealing with minor variations of the original model, but without changing its main macroscopic properties): let us cite \cite{AGLS10, BBCS, BESS, BES, DES24, O, GKR, Gold21, Gold2019,  GLS2024, ZC2019} for the one-dimensional ($d=1$) case, and \cite{ERSS24, GLS25, HL, Lubeck, RPV} for the higher $d\ge 2$ case. The one-dimensional system has very peculiar features which makes the model mathematically tractable, and all the convergence results in the above cited papers are rigorously proved. When the dimension is higher, its rigorous treatment is more challenging: numerical simulations suggest that it exhibits a hyperuniform critical state, similarly to several related models featuring an absorbing phase transition. To our knowledge, no theoretical result has been mathematically proved so far.

The purpose of this paper is to present some of the arguments used by physicists in order to study the near critical scaling of the model, in a language more approachable to mathematicians. In particular, we will derive with more details the relationships between critical exponents which have been stated in \cite{ERSS24}. We emphasize that, while targeted at readers from the mathematical community, the results presented are mostly non-rigorous (with the exception of Sections \ref{sec:stationary} and \ref{sec:d1}). For further background and motivation we refer the reader to \cite{ERSS24} and references therein. 

\subsection*{Overview of the paper}In Section \ref{sec:model} we introduce the CLG model, and we make some preliminary observations in dimension $d=2$, then in Section \ref{sec:macro} we define the macroscopic observables and their critical exponents. In Section \ref{sec:scaling} we derive the scaling relations which have been stated in \cite{ERSS24}. In Section \ref{sec:stationary} we add a boundary dynamics and obtain rigorously the expression of the stationary current. Finally, in Section \ref{sec:d1} we review several expressions for the macroscopic observables which have been obtained explicitely in the case $d=1$.

\subsection*{General notations}
 We use $\llbracket$ and $\rrbracket$ as delimiters for integer segments, meaning for example that  given two integers  $a<b$,  $\llbracket a,b \rrbracket =\{a,a+1,\dots,b\}$.

\section{The model}\label{sec:model}

Let $d\ge 1$ and $L\in \N$ be two positive integers, and consider the $d$-dimensional periodic lattice $\T_L:=\llbracket 1,L\rrbracket^d$ with $0\equiv L$. Two sites $i,j \in \T_L$ are called neighbours if $\|i-j\|_1=1$ and we denote it by $i\sim j$. We introduce the space $\Sigma_L:=\{0,1\}^{\T_L}$ which is the space of \emph{particle configurations}: indeed, we denote a configuration by $\eta:=(\eta_i)_{i\in \T_L}$, where $\eta_i=1$ (resp.~$\eta_i=0$) indicates that site $i$ is occupied by a particle (resp.~empty). The \emph{$d$-dimensional conservative lattice gas (CLG)}  is driven by the following Markov generator
\eqs{\mathscr{L} f(\eta):=\sum_{i\in \T_L}\sum_{j\sim i}(1-\eta_j)\eta_i\;{\bf1}_{\big\{\sum_{j'\sim i}\eta_{j'}>0\big\}}\big\{f(\eta^{i,j})-f(\eta)\big\}}
where $\eta^{i,j}$ designates the configuration identical to $\eta$ except that the particle at site $i$ has jumped to $j$.
In other words, a particle at site $i$ jumps at rate $1$ to any empty neighbouring site $j\sim i$,  if and only if another neighbouring site $j'\sim i$ is occupied (this is the \emph{kinetic constraint}). The particles with no neighbouring particle are called \emph{frozen} particles, while  particles which have at least one other neighbouring particle are called \emph{active} particles. See Figure \ref{fig:1} for an illustration of a configuration.

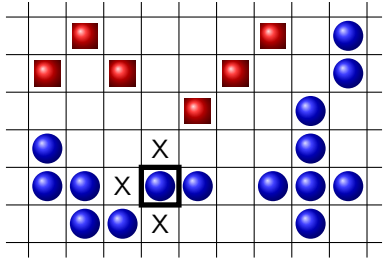
\begin{figure}
	\begin{tikzpicture}
		\draw (0.2,0) -- (5.2,0);
		\draw (0.2,0.5) -- (5.2,0.5);
		\draw (0.2,1) -- (5.2,1);
		\draw (0.2,1.5) -- (5.2,1.5);
		\draw (0.2,2) -- (5.2,2);
		\draw (0.2,2.5) -- (5.2,2.5);
		\draw (0.2,3) -- (5.2,3);
		\draw (0.5,-0.2) -- (0.5,3.2);
		\draw (1,-0.2) -- (1,3.2);
		\draw (1.5,-0.2) -- (1.5,3.2);
		\draw (2,-0.2) -- (2,3.2);
	%	\pattern[pattern=north east lines,opacity=0.5] (0.2,1.5)--(4,1.5)--(4,3.3)--(0.2,3.3)--cycle;
		\shade[ball color=blue](1.25,0.75) circle (0.2);
		\shade[ball color=blue](0.75,0.75) circle (0.2);
		\shade[ball color=blue](0.75,1.25) circle (0.2);
		\shade[ball color=red](1.58,2.08) rectangle (1.92,2.42);
		\shade[ball color=blue](2.75,0.75) circle (0.2);
		\shade[ball color=red](3.08,2.08) rectangle (3.42,2.42);
		\shade[ball color=red](3.58,2.58) rectangle (3.92,2.92);
		\shade[ball color=blue](4.25,1.75) circle (0.2);
		\shade[ball color=red](0.58,2.08) rectangle (0.92,2.42);
		\shade[ball color=blue](1.75,0.25) circle (0.2);
		\shade[ball color=blue](1.25,0.25) circle (0.2);
		\shade[ball color=blue](2.25,0.75) circle (0.2);
		\draw[line width=0.7mm] (2,0.5) rectangle (2.5,1);
		\shade[ball color=red](2.58,1.58) rectangle (2.92,1.92);
		\shade[ball color=red](1.08,2.58) rectangle (1.42,2.92);
		%	\shade[ball color=blue](1.75,2.75) circle (0.2);
		\shade[ball color=blue](4.75,2.75) circle (0.2);
		\shade[ball color=blue](4.75,2.25) circle (0.2);
		\shade[ball color=blue](4.25,0.75) circle (0.2);
		\shade[ball color=blue](4.25,1.25) circle (0.2);
		\shade[ball color=blue](4.25,0.25) circle (0.2);
		\shade[ball color=blue](3.75,0.75) circle (0.2);
		\shade[ball color=blue](4.75,0.75) circle (0.2);
		
		\node at (2.25,1.25) {\textsf{X}};
		\node at (2.25,0.25) {\textsf{X}};
		\node at (1.75,0.75) {\textsf{X}};
		%	\centerarc[<-](1.5,0.85)(10:170:0.25);
		
		%	\centerarc[->](1.25,1.05)(240:100:0.25);
		
		%	\centerarc[<-](1.25,0.4)(270:120:0.25);
		
		%	\draw[] (2.75,-0.5) node {\small \myemph{Active} \textbf{particles}};
		%		
		%		\draw (7.2,0) -- (9.2,0);
		%		\draw (7.2,0.5) -- (9.2,0.5);
		%		\draw (7.2,1) -- (9.2,1);
		%		\draw (7.2,1.5) -- (9.2,1.5);
		%		\draw (7.5,-0.2) -- (7.5,1.7);
		%		\draw (8,-0.2) -- (8,1.7);
		%		\draw (8.5,-0.2) -- (8.5,1.7);
		%		\draw (9,-0.2) -- (9,1.7);
		%		\shade[ball color=red](8.25,0.75) circle (0.2);
		%		%\shade[ball color=green](0.75,0.75) circle (0.2);
		%		\draw[] (8.25,-0.5) node {\small \myalert{Frozen} \textbf{particles}};
		%		
		
		\draw (2.5,-0.2) -- (2.5,3.2);
		\draw (3,-0.2) -- (3,3.2);
		\draw (3.5,-0.2) -- (3.5,3.2);
		\draw (4,-0.2) -- (4,3.2);
		\draw (4.5,-0.2) -- (4.5,3.2);
		\draw (5,-0.2) -- (5,3.2);

	\end{tikzpicture}
	\caption{{\color{blue}Blue circle} particles are active, {\color{red}red square} particles are frozen. As an example, the active particle highlighted with $\square$~jumps at rate $1$ to one of its three neighbours indicated with \textsf{X}.}
	\label{fig:1}
\end{figure}

To illustrate the macroscopic behavior of this model, let us now focus on the case $d=2$. Because of the kinetic constraint, the system can \textit{freeze} (\textit{i.e.}~no particle can move) if the number of particles $n$ is less than roughly $L^2/2+\mathcal{O}(L)$ (depending on $L$'s parity). In this case, the CLG ultimately reaches an alternate chessboard-like configuration $\arraycolsep=0.15pt\def\arraystretch{0.05}\begin{array}{ccc}
\bullet & \circ&\bullet\\
\circ&\bullet&\circ\\
\bullet & \circ&\bullet
\end{array}$
(where $\bullet$ and $\circ$ respectively represent occupied and empty sites).

 On the other hand, for $n\gtrsim L^2/2$, no such alternate configuration exists, and the CLG remains active forever. However, as proved numerically in \cite{ERSS24}, there exists an (asymptotic) critical density $\rho_c<1/2$, such that at density $\rho:=n/L^2\in (\rho_c, 1/2)$, starting from some random configuration, activity typically does not vanish over diffusive timescales $ t=\mathcal{O}(L^2)$. Roughly speaking, reaching a frozen configuration for such densities requires overpowering a potential barrier, as many more jumps move the system away from an alternate configuration than towards it. In this regime, the average transience time to reach a frozen state (started from a uniform configuration with $n$ particles, for example) increases exponentially with the size of the system. Over diffusive timescales, we expect that the system locally reaches a \emph{quasi-stationary state} $\pi_\rho$ parametrized by the density $\rho$. In the permanently active phase $\rho>1/2$, we will denote by $\pi_\rho$ the (exact) stationary state at density $\rho$. We expect that this picture holds true for any $d \geq 2$. Note that for $d=1$, it is established that $\rho_c=1/2$ \cite{BESS,BES,DES24}. This special case is mathematically very well understood, due to some pecularities of the one-dimensional dynamics:  in particular the exact stationary state $\pi_\rho$, $\rho>1/2$, is completely explicit, and can be described either by its finite size marginals \cite{BESS}, or by a Markovian construction  \cite{DES24}.

\section{Macroscopic observables} \label{sec:macro}
We now describe the various macroscopic observables pertaining to the CLG that will yield the relations between scaling exponents laid out in \cite{ERSS24} in arbitrary dimension $d$.
\subsection{Active density and activity}
In accordance with the kinetic constraint, we say that a particle is \emph{active} if it has at least one occupied neighbouring site (see also Figure \ref{fig:1}). Set, for any $i\in\T_L$, 
\eq{eq:defAi}{A_i:=\begin{cases}1&\mbox{ if there is an active particle at site }i\\ 0& \mbox{ otherwise,}\end{cases}} and define the number of active particles in the system as 
\eqs{n_a:=\#\{i\in \T_L,\; A_i=1\}.}
We denote by $\rho_a:=n_a/L^d$ the corresponding (empirical) active density. Note that under the translation invariant quasi-stationary state $\pi_\rho$, we can identify $\rho_a(\rho)$ as the quasi-stationary expectation $\E_{\pi_\rho}[{A_i}]$ computed at an arbitrary site $i$. As the density goes down to the critical density $\rho_c$, it is shown numerically (see \cite{ERSS24, RPV} when $d=2$) that the active density vanishes polynomially, and we set 
\eq{eq:scalingrhoa}{\rho_a(\rho)\sim (\rho- \rho_c)^\beta \qquad \text{as }  \rho\to\rho_c.}
Dense regions typically contain a lot of active particles, even though very few jumps are allowed due to the exclusion constraint. For this reason, we also define  the number of edges over which a jump of an active particle can actually occur, towards an empty site,  
\eqs{\mathfrak{a}:=\#\{i\sim j\in \T_L,\; A_i(1-\eta_j)=1\}.}
and denote by $a(\rho):=\mathfrak{a}/L^d\leq (2^d-1)\rho_a(\rho)$ the associated \emph{activity}. Once again, the activity vanishes polynomially \cite{ERSS24} as the density goes down to $\rho_c$, and we set
\eqs{a(\rho)\sim (\rho- \rho_c)^b  \qquad \text{as }  \rho\to\rho_c.}
In \cite{ERSS24} it is shown numerically that $b=\beta$ in the case $d=2$. We conjecture that the same holds in any larger dimension, because at the critical density, most chains of active clusters are likely pairs of neighbouring particles with all other neighbours being empty, resulting in $a(\rho)\simeq (2^d-1)\rho_a(\rho)$ for $\rho\simeq \rho_c$. This is verified numerically in $d=2$.

\subsection{Diffusion coefficient}
In a configuration $\eta$, the  instantaneous current along any edge $(i,j)$ in the system is given by $J_{i,j}(\eta)=A_i-A_j$, which takes the form of a discrete gradient. For this reason, recalling that $\rho_a(\rho) = \E_{\pi_\rho}[A_i]$, we expect that a law of large numbers holds as the lattice size diverges (see for instance \cite{KL}), and that the system under diffusive rescaling obeys the following hydrodynamic equation on the continuous $d$-dimensional torus $\T:=[0,1)^d$
\eq{eq:HDL}{\partial_t \rho=\Delta \rho_a(\rho),} where $\rho:\R_+\times \T \to [0,1]$ is the \emph{macroscopic density} of particles. 
Provided one can prove rigorously the hydrodynamic limit, the function $\rho$ in \eqref{eq:HDL} can be seen as the weak limit of the empirical density, meaning that for any smooth test function $H$ on $\T$,
\eq{eq:conv}{\frac{1}{L^d}\sum_{i\in\T_L}H(i/L) \eta_i(tL^2) \xrightarrow[L\to+\infty]{} \int_{\T} H(x) \rho(t,x) dx,} where the limit usually holds in probability.
The associated macroscopic diffusion coefficient  is therefore given by
\eq{eq:defD}{D(\rho):=\frac{\rm{d}}{\rm{d}\rho}\rho_a(\rho),} 
for which
$\partial_t \rho=\nabla \cdot (D(\rho) \nabla \rho)$. Similarly, close to criticality, we expect 
\eq{eq:defDbis}{D(\rho)\sim (\rho- \rho_c)^\alpha \qquad \text{as } \rho \to \rho_c,}
which together with \eqref{eq:scalingrhoa} yields the first scaling relation \cite[Equation (7)]{ERSS24}, namely
\eq{eq:scaling1}{\alpha=\beta-1.}

\subsection{Correlation lengths}
In the supercritical regime $\rho>\rho_c$, the CLG is characterized by two distinct correlation lengths, which are generically defined as follows:

\begin{itemize} 

\item {\bf The 2-points correlation length $\xi_\times >0$,} classical from a probabilistic standpoint,  is the inverse of the (commonly observed) exponential decay rate of the two-points correlation function given for any $i,j \in \T_L$ by
\eq{eq:2PCF}{\varphi_\rho(i,j):=\E_{\pi_\rho}\big[(\eta_i-\rho)(\eta_j-\rho)\big]\sim \exp\pa{-\frac{\|i-j\|_1}{\xi_\times}} \qquad \text{as }\rho \to \rho_c.}

\item {\bf The geometric correlation length $\xi_\perp>0$} is the most used in the physics literature, see for instance \cite{RPV}. It relates to the spread of activity, meaning that a typical active ``cluster'' (\textit{i.e.}~a group of particles) is characterized by a self-sustained chain of activation among particles of the cluster, and   $\xi_\perp$ is then defined as the typical length-scale of these self-sustained clusters. 
In particular, in a closed system of size $L \ll \xi_\perp$, self-activation becomes impossible, hence activity quickly dies out. However, when $L \gg \xi_\perp$ there is enough space for clusters to self-sustain, making transience time extremely long. 
%In particular, the local transience time of a system at density $\rho\in(\rho_c,1/2)$ will be the same in a closed (non-periodic) system of size $L\gg \xi_\perp$ as in an infinite one, meaning activity will sustain itself over long timescales, whereas in a closed system of size $L\ll \xi_\perp$, activity quickly dies out.

\end{itemize}

It is generically observed that both  correlation lengths diverge polynomially as the density goes down to $\rho_c$, and as in \cite{ERSS24} we denote by $\nu_\times$ and $\nu_\perp$ the corresponding critical exponents, namely
\eqs{\xi_\times(\rho)\sim (\rho- \rho_c)^{-\nu_{\times}} \qquad   \mbox{ and }  \qquad \xi_\perp(\rho)\sim (\rho- \rho_c)^{-\nu_{\perp}} \qquad \text{as } \rho \to \rho_c.}

\subsection{Compressibility}
For interacting lattice gases, the compressibility is usually defined as the summed two-points correlation in the infinite system, namely
\eq{eq:defChi}{\chi(\rho):=\sum_{i\in\Z^d}\varphi_\rho(0,i),}
where $\varphi_\rho$ was defined in \eqref{eq:2PCF}.
We denote by $\gamma$ the critical exponent for the compressibility
\eq{eq:defchi}{\chi(\rho) \sim (\rho- \rho_c)^{\gamma} \qquad \text{as } \rho \to \rho_c.}

\subsection{Conductivity}
There are several equivalent ways to define the density-dependent conductivity $\sigma(\rho)$ for interacting lattice gases. One such a way which is particularly relevant in our case is through the CLG's macroscopic fluctuation field $u=\rho-\rho_c$, defined for a quasistationary system as the limit in probability:
\eqs{\frac{1}{L^{d/2}}\sum_{i\in\T_L} H(i/L) (\eta_i(tL^2)-\rho_c)\xrightarrow[L\to+\infty]{} \int_\T H(x) u(t,x)dx.}
This fluctuation field is expected to be a solution to 
\eq{eq:fluctuationsHL}{\partial_t u = \nabla \cdot (D(\rho)\nabla u) + \sqrt{2\sigma(\rho)}W,}
where $W$ is a space-time stationary field, which behaves differently according to the space scaling considered: far away from criticality, or when $x\gg \xi_\times$, fluctuations decorrelate and 
$W$ is a standard space-time white noise; close to criticality, and at distances  $x\ll \xi_\times$ instead,  $W$ has non-trivial space correlations of the form
\eq{eq:spacecorr}{\E\big[W(t,x)W(0,0)\big]=\delta(t)|x|^{-\theta}.}
The noise $W$ should then be interpreted as a random distribution, whose formal integral 
\begin{equation}
\label{eq:defW}
W(f):=\int_{\R\times \R^d}f(t,x)W(t,x)dtdx
\end{equation}
is a centered gaussian for any function $f:\R\times \R^d\to \R$, and  \eqref{eq:spacecorr} identifies its covariance as
\begin{equation}
\label{eq:covW}
\E\big[W(f)W(g)\big]=\int_{\R}\cro{\iint_{\R^d\times \R^d}\frac{f(t,x)g(t,y)}{|x-y|^\theta } dx dy } dt.
\end{equation}
For gradient, nearest-neighbour, isotropic exclusion processes like the CLG, the conductivity $\sigma(\rho)$ appearing in \eqref{eq:fluctuationsHL} can be expressed as the average equilibrium rate at which any given edge in the system is crossed in either direction, namely for $0=(0,\dots,0)$ and  $e_1=(1,0,\dots,0)$ 
\[\sigma(\rho)=\E_{\pi_\rho}\big[c_{0,e_1}(\eta)\big],\]
where
\begin{equation}
\label{eq:cij}
c_{i,i'}(\eta):=(1-\eta_{i'})A_i+(1-\eta_i)A_j.
\end{equation}
Note that the conductivity also relates to the macroscopic response of the model to an external field \cite{Spohn}, however we will not give more details on this equivalent interpretation of $ \sigma(\rho)$.

\subsection{Local density fluctuations}

The CLG displays \emph{self-organized criticality}: in other words,  letting a large cluster full of particles spread out in an empty infinite system until it freezes makes the local density inside the cluster converge to $\rho_c$. In dimension $d=2$, it has been observed \cite{ERSS24, GLS25, HL} that at the critical density $\rho_c$, the system ultimately reaches a \emph{hyperuniform} frozen state, meaning  that  the standard deviation of the number of particles in a box of size $R$ in a typical frozen state is of order $R^{\zeta}$ for some $\zeta\leq d/2$. That is,
\begin{equation}
\var(\text{number of particles in }\llbracket R \rrbracket^d)\sim R^{2\zeta}.
\end{equation}

\section{Scaling relations} \label{sec:scaling}

We are now ready to derive the scaling relations which have been stated in \cite{ERSS24}.

\subsection{Compressibility at criticality}
Given a positive integer $k$, we denote   $\xi_k:=k\xi_\times$ and $\xi_\ell:=(k-\sqrt k)\xi_\times$. We introduce two boxes:  $\Lambda_m:=\llbracket-\xi_m, \xi_m\rrbracket^d$, for $m=k,\ell$. Since at distance much larger than $\xi_\times$ the system decorrelates, for large $k$, and $\rho$ close to $\rho_c$, we can write
\begin{align*}
\var_{\pi_\rho}\bigg(\sum_{i\in  \Lambda_k}\eta_i\bigg)&=\sum_{i,j\in \Lambda_k}\varphi_{\rho}(i,j)\\
&=\sum_{i\in \Lambda_{\ell}}\sum_{j\in \Lambda_k}\varphi_{\rho}(i,j)+\sum_{i\in \Lambda_k\setminus \Lambda_\ell}\sum_{j\in \Lambda_k}\varphi_{\rho}(i,j).
\end{align*}
Assuming that  $\sum_{i}\varphi_{\rho}(0,i)$ converges absolutely, and since for large $k$, by definition of the two-points correlation length, $ \chi(\rho) \simeq \sum_{i\in  \Lambda_{\sqrt{k}}}\varphi_\rho(0,i)$, we obtain that for any $i\in \Lambda_\ell$, 
$\sum_{j\in \Lambda_k}\varphi_{\rho}(i,j)\simeq \chi(\rho)$, so that
\eqs{\var_{\pi_\rho}\bigg(\sum_{i\in  \Lambda_k}\eta_i\bigg)=|\Lambda_\ell| \chi(\rho)+\mathcal{O}(|\Lambda_k\setminus\Lambda_\ell|)=|\Lambda_k| \chi(\rho)+\mathcal{O}(k^{3/2}\xi_{\times}^d).}
This yields $\xi_k^{2\zeta}=4\xi_k^d \chi(\rho)+O(\xi_k^d/{k^{d-3/2}} )$. Dividing by $k^d$, for large but fixed $k$, when obtain as $\rho\to\rho_c$, the identity
\eqs{\gamma=\nu_\times(d-2\zeta),}
which is  \cite[Equation (1)]{ERSS24}.

%\mm{check if this is correct in dimension $d$}

\subsection{Einstein's relation}
Einstein's fluctuation-dissipation relation
\eq{Einstein}{\sigma(\rho)=D(\rho)\chi(\rho)}
relates the two constants characterizing the large scale evolution of the model, and leads to the following relation 
\begin{equation}
\label{eq:EinsteinExponents}
\alpha=b-\gamma,
\end{equation}
which is \cite[Equation (4)]{ERSS24}. To sketch its justification in our setting, we follow similar arguments to those presented in \cite[Section II.2]{Spohn}. We define, in infinite volume, the space-time correlation function in equilibrium 
\eq{eq:2PCFt}{\psi_\rho(t,i):=\E_{\pi_\rho}\big[(\eta_i(t)-\rho)(\eta_0(0)-\rho)\big]}
Note in particular that $\psi_\rho(0,i)=\varphi_\rho(0,i),$ where $\varphi_\rho$ is the two-points correlation function defined in \eqref{eq:2PCF}. As a consequence, the compressibility is simply interpreted as the total mass of $\psi_\rho(0,\cdot)$,
\[\chi(\rho)=\sum_{i\in\Z^d}\psi_\rho(0,i).\]
Because mass is conserved by the dynamics, so is the summed correlation, and because of the decay of correlations, away from criticality, this mass is  initially concentrated on a local neighbourhood of the origin. Furthermore, according to the equation ruling the fluctuation field \eqref{eq:fluctuationsHL}, $\psi_\rho$ should be solution to the discrete heat equation, with diffusion coefficient $D(\rho)$
\[\partial_t \psi_\rho(t,i)=D(\rho) \Delta \psi_\rho(t,i),\] 
where the Laplacian above is the discrete one, namely $\Delta f(i) = \sum_{j\sim i}(f(j)-f(i))$.
In particular, for any time $t$,
\[\sum_{i\in\Z^d} i_1^2\;[\psi_\rho(t,i)-\psi_\rho(0,i)]= t\chi(\rho)D(\rho),\]
because when time-differentiating the right-hand side, and summing by parts over $i_1$, we obtain $D(\rho)\sum_i \psi_\rho(t,i)=D(\rho)\chi(\rho)$. We now need to prove that  
\[\frac{1}{t}\sum_{i\in\Z^d} i_1^2\;[\psi_\rho(t,i)-\psi_\rho(0,i)]=\frac{1}{t}\sum_{i\in\Z^d} i_1^2\;\E_{\pi_\rho}\Big[(\eta_i(t)-\eta_i(0))\eta_0(0)\Big]=\sigma(\rho).\]
By stationarity and translation invariance of $\pi_\rho$, the quantity above is 
\[-\frac1{2t}\sum_{i\in\Z^d} i_1^2\;\E_{\pi_\rho}\Big[(\eta_i(t)-\eta_i(0))(\eta_0(t)-\eta_0(0))\Big].\]
We now write, for any $i$, 
\begin{equation}\label{eq:etai}\eta_i(t)-\eta_i(0)=\int_{0}^t \sum_{i'\sim i}J_{i,i'}(s) ds +\sum_{i'\sim i}M_t^{i,i'}=\int_{0}^t \Delta A_i(s) ds +\sum_{i'\sim i}M_t^{i,i'}\end{equation}
for a family of martingales satisfying 
\begin{equation}
\label{marteq}
\E_{\pi_\rho}\big[(M_t^{i,i'})^2\big]=-\E_{\pi_\rho}\big[M_t^{i,i'}M_t^{i',i}\big]=t\E_{\pi_\rho}\big[c_{i,j}\big]=t\sigma(\rho).
\end{equation}
and $\E[M_t^{i,i'}M_t^{j',j}]=0$ if $\{i,i'\}\neq \{j,j'\}$. Once again, in the identity above, $\Delta$ represents the lattice Laplacian.

\medskip
Now, let us sum \eqref{eq:etai} over $i$, and integrate against $\eta_0(t)-\eta_0(0)$, we see that the contribution of the integral part vanishes. We are left with
\begin{multline*}
-\frac1{2t}\sum_{i\in\Z^d} i_1^2\;\E_{\pi_\rho}\cro{\bigg(\sum_{i'\sim i}M_t^{i,i'}\bigg)\bigg(\int_{0}^t \Delta A_0(s) ds +\sum_{i'\sim 0}M_t^{0,i'}\bigg)} \\
=-\frac1{2t}\sum_{i\in\Z^d} i_1^2\;\E_{\pi_\rho}\cro{\bigg(\sum_{i'\sim i}M_t^{i,i'}\bigg)\pa{\int_{0}^t \Delta A_0(s) ds}}-\frac1{2t}\sum_{i\sim 0}i_1^2\;\E_{\pi_\rho}\cro{M_t^{0,i}M_t^{i,0}},
\end{multline*}
since the cross product for $i=0$ vanishes because of the factor $i_1$, and cross products of martingales vanish as soon as they do not involve the same edge. 
Regarding the second term, 
\[-\frac1{2t}\sum_{i\sim 0}i_1^2\;\E_{\pi_\rho}\cro{M_t^{0,i}M_t^{i,0}}=-\frac1{2t}\E_{\pi_\rho}\cro{M_t^{0,e_1}M_t^{e_1,0}}=\sigma(\rho).\]
Regarding the first term, on the other hand, we define $M_t^i:=\sum_{i'\sim i}M_t^{i,i'}$, and use the translation invariance of the stationary state, to obtain after a summation by parts
\begin{align*}
-\frac1{2t}\sum_{i\in\Z^d} i_1^2\;\E_{\pi_\rho}\cro{\bigg(\sum_{i'\sim i}M_t^{i,i'}\bigg)\pa{\int_{0}^t \Delta A_0(s) ds}}&=-\frac1{2t}\sum_{i\in\Z^d} i_1^2\;\E_{\pi_\rho}\cro{\Delta M_t^i\int_{0}^t A_0(s) ds} \\
&=-\frac1{t}\E_{\pi_\rho}\cro{\int_{0}^t A_0(s) ds \sum_{i\in \Z^d}  M_t^i}\\
&=0,
\end{align*}
since $\sum_i  M_t^i$ vanishes because the total current along an edge and its inverse cancel out.
We finally obtain Einstein's relation \eqref{Einstein}.

\subsection{Hiding the off-criticality}
The geometric correlation length $\xi_{\perp}$ can  be thought of as the scaling starting which we start ``feeling" the off-criticality.

Consider a box of length $L$. The empirical density $\rho$ of particles in this box is never exactly $\rho_c$; let us assume $\rho > \rho_c$. The deviation $\rho - \rho_c$ could be explained in two ways: either the system is truly off-critical, or it is critical but random fluctuations inside our box increased slightly the measured density.
In fact, if $L^d(\rho - \rho_c) \lesssim L^\zeta$ we will be unable to determine which explanation is the correct one.

Imagine now that the system was prepared at an off-critical density $\rho > \rho_c$. Up to the scale $L$ satisfying $L^d(\rho-\rho_c) \sim L^{\zeta }$, an observer might (mistakenly) think that the system is critical, so, for example, if we close the boundaries of the box the transience time will be short. At scales above $L$, however, the observer knows that we are in a supercritical phase, and in particular the transience time is very long.

We thus see that the geometric length scale $\xi_\perp$ can be identified with $L$ which satisfies $L^d(\rho-\rho_c) \sim L^{\zeta }$, yielding the relation
\begin{equation}
\nu_\perp (d-\zeta)=1, \label{eq:nuperp_zeta}
\end{equation} 
which is  \cite[Equation (3)]{ERSS24}.

\subsection{Scaling invariance and dynamical exponent}

We denote $u=\rho-\rho_c$, and identify functions of $\rho$ as functions of $u$ (\emph{e.g.}~$D(\rho)$ as $D(u)$). Close to criticality we have $D(u)\sim u^\alpha = u^{\beta-1}$ and $\sigma(u)\sim u^{\gamma+\beta-1}$.
%Moreover in dimension $2$ we have $\nu_\perp(2-\zeta)=1$ and $2\nu_\times(1-\zeta)=1$.
In the regime $x\ll \xi_\times$, where spatially correlated noise affects the fluctuation field,  we define the rescaled fluctuation field, with scaling parameter $\ell >0$, as
\[ \tilde u(x,t)=\ell^{\zeta-d} u(t/\ell^z,x/\ell).\]   
We look for a relationship between exponents which make equation \eqref{eq:fluctuationsHL} invariant by this rescaling. Recall that $W$ in this regime has space correlations given by \eqref{eq:spacecorr}.
Easy computations give
\begin{align*}
\partial_t \tilde u(t,x)&=\ell^{\zeta-d-z} \partial_t u(t/\ell^z,x/\ell)\\
\nabla \tilde u(t,x) & = \ell^{\zeta-d-1} \nabla u (t/\ell^z,x/\ell)\\
D(\tilde u) & \sim \ell^{(\zeta-d)(\beta-1)} (D(u))(t/\ell^z,x/\ell)\\
%D(\tilde u)\nabla \tilde u &\sim \underbrace{\ell^{(\zeta-d)(\beta-1)+\zeta-d-1}}_{=\ell^{(\zeta-d)\beta-1}} \big( D(u)\nabla u\big)(x/\ell,t/\ell^z)\\
\nabla\cdot \big(D(\tilde u)\nabla \tilde u\big) &= \ell^{(\zeta-d)\beta-2} \big(\nabla\cdot  \big(D(u)\nabla u\big)\big)(t/\ell^z,x/\ell)\\
\sqrt{\sigma(\tilde u)} & = \ell^{\frac {(\zeta-d)\gamma} 2 + \frac {(\zeta-d)} 2(\beta-1)} \sqrt{\sigma(u)}(t/\ell^z,x/\ell)
\end{align*}
Furthermore, if we define for some constant $a$ (to be chosen later), 
\[\tilde W(t,x)dt dx = \ell^{a} W(t/\ell^z,x/\ell)dt dx,\]
 then equation \eqref{eq:defW} rewrites as 
\[\tilde W (f)=W(\tilde f) \qquad \text{for } \tilde f(t,x)=\ell^{a+d+z}f(\ell^zt, \ell x)\]
and \eqref{eq:covW} yields
\[\E\big[\tilde W (f) \tilde W (g)\big]=\E\big[W(\tilde f)W(\tilde g)\big]=\ell^{2a+z+\theta}\E[W(f)W(g)],\]
so that in order to enforce scale invariance of the noise, we choose $ a=-(z+\theta)/2$.

From the above we can now write
\[\nabla\big(\sqrt{\sigma(\tilde u)}\tilde W\big) = \ell^{\frac{(\zeta-d)(\beta+\gamma-1)}{2}-\frac{z+\theta}{2}-1} \nabla\big(\sqrt{\sigma(u)}W\big)(x/\ell,t/\ell^z).\]
Therefore, scale invariance of equation \eqref{eq:fluctuationsHL} is ensured when
\[ \zeta-d-z=(\zeta-d)\beta-2=\frac{(\zeta-d)(\beta+\gamma-1)}{2}-\frac{z+\theta}{2}-1.\]
which is solved by 
\[ z=(\zeta-d)(1-\beta)+2, \quad \text{and} \quad \theta=d-\zeta,\]
when $\gamma =1$ (which corresponds to the case $\beta=b$, see \eqref{eq:scaling1} and \eqref{eq:EinsteinExponents}).
This yields  \cite[Equation (11)]{ERSS24}. 

\medskip

Note that  in the scale $x\gg \xi_\times$, the invariance principle is satisfied with the exponents $z=2,\zeta=\frac{d}{2},\gamma=0,\beta=1$, taking $\theta=d$ giving same scaling as white noise (and is the smallest value of $\theta$ for which the covariance is integrable over space).

\section{Boundary-driven CLG and stationary current} \label{sec:stationary}
In the non-periodic setting, in the presence of reservoirs, the system reaches a stationary state (in general non-reversible), which we assume to locally resemble the quasi-stationary state of the closed system. We consider here the non-periodic box $\Lambda_L=\llbracket 1,L\rrbracket^d$, whose boundary 
\[\partial \Lambda_L:=\{i\in \Lambda_L \; ;\;  \exists j\notin \Lambda_L:\; j\sim i\}\]
can be equipped with several choices of non-conservative dynamics. In \cite{DES24}, reservoirs mimicking contact with infinite FEP outside $\Lambda_L$ are considered. We, however, choose here simpler boundary dynamics which resamples a Benoulli random variable at each boundary point, corresponding to an unconstrained heat bath outside $\Lambda_L$.
%Here, we do not choose the physically relevant one, namely the one that mimics contact with infinite FEP reservoirs studied in \cite{DES24}. Instead, we choose here the simplest possible boundary dynamics which resamples a Benoulli random variable at each boundary point.
This dynamics is driven by the generator
\begin{align}
&\mathscr{L} = \mathscr{L}_{bulk} + \mathscr{L}_{boundary}, \vphantom{\bigg(}\\
\mathscr{L}_{bulk} f(\eta):=\sum_{i\in \Lambda_L}&\sum_{j\in \Lambda_L,\, j\sim i}A_i(1-\eta_j)\big\{f(\eta^{i,j})-f(\eta)\big\}\label{eq:boundary}\\
\mathscr{L}_{boundary} f(\eta) := \sum_{i\in \partial \Lambda_L}& \left( \alpha(i) f(\eta^{i \leftarrow 1}) +(1-\alpha(i)) f(\eta^{i \leftarrow 0})- f(\eta)\right),
\end{align}
where $A_i=\eta_i\;{\bf1}_{\{\sum_{j'\sim i}\eta_{j'}>0\}}$ defined in \eqref{eq:defAi} indicates whether there is an active particle at site $i$, and $\alpha:\partial\Lambda_L \to (0,1)$ represents the boundary active density.
In order for the dynamics to be ergodic, we set boundary particles to be always active, i.e.~$\eta\equiv 1$ in $\Lambda_L^{\mathrm{c}}$.
With this choice of boundary dynamics, due to the \emph{gradient property} of the model (see \cite{KL} for instance), the density of active particles in stationarity solves a Dirichlet problem.

\medskip

In order to formulate the result, we introduce a  \emph{mirror boundary}, \[\partial^*\Lambda_L=\{i\in \Z^d \; ; \; d(i,\Lambda_L)=1\}\] and define on this mirror boundary $\rho_a(i^*) = \alpha(i)$ for $i^*\in\partial^*\Lambda_L$.
In addition, we denote the unique stationary measure by ${\pi}_\alpha$, and define similarly as before, for $i\in  \Lambda_L$, 
\[\rho_a(i) = \E_{\pi_\alpha}[A_i].\]

\begin{theorem} The function
 $\rho_a$ solves the Dirichlet problem:
\[\begin{cases}
\sum_{j \sim i} (\rho_a(j)-\rho_a(i)) = 0, \quad  \forall i\in  \Lambda_L\\
\rho_a(i^*) = \alpha(i),  \quad \forall i^*\in \partial^*\Lambda_L,
\end{cases}
\]
 where the sum in the first line is taken over $j\in \Lambda_L\cup \partial^*\Lambda_L$ neighbouring $i$.

Moreover, the expected current over an edge $i\sim j$ is given by $\rho_a(j)-\rho_a(i)$, and when $j=i^*$ this current should be seen as the rate of particles entering from the reservoir via the boundary site $i$ (with exiting particles counted with a negative sign).
\end{theorem}
\begin{proof}
Fix $i\in\Lambda_L$. By stationarity,  $ \pi_\alpha(\mathscr{L} \eta(i) )=0$.
In the bulk,
\[
\mathscr{L}_{bulk} \eta(i) = \sum_{\Lambda_L \ni j\sim i} (A_j - A_i).
\]
If $i\in \partial \Lambda_L$,
\[
\mathscr{L}_{boundary} \eta(i) = \rho(i^*) - A_i = \sum_{\partial^*\Lambda_L \ni j \sim i} (\rho(j)-A_i).
\]
Taking expectation with respect to $\pi_\alpha$ and summing both contributions proves the theorem.
% If $i\notin  \Lambda_L$, this implies
%\[
%\E_{\pi_\alpha} \bigg[\sum_{j\sim i} (A_j-A_i) \bigg] = 0,
%\]
%so $\sum_{j \sim i} (\rho_a(j)-\rho_a(i)) = 0$.
\end{proof}

As a straightforward consequence, we have the following results.
\begin{corollary}
Consider a cylindrical geometry in two dimensions, $\Lambda = [1,L]\times \mathbb{Z}/{L\mathbb{Z}}$, with boundary $\partial \Lambda = \{1,L\}\times \mathbb{Z}/{L\mathbb{Z}}$. We set the left reservoir to  $\alpha_\ell$ and the right reservoir to $\alpha_r$, i.e.,
\begin{align*}
\alpha(i) &= \alpha_\ell \quad  \text{if } i \in \{1\}\times \mathbb{Z}/{L\mathbb{Z}},\\
\alpha(i) &= \alpha_r  \quad \text{if } i \in \{L\}\times \mathbb{Z}/{L\mathbb{Z}}.\\
\end{align*}
Then we have
\[
\rho_a (i_1,i_2) = \alpha_\ell + (\alpha_r-\alpha_\ell)\frac{i_1}{L+1}.
\]
\end{corollary}
\begin{corollary}
Consider a cylindrical geometry in two dimensions as previously, and let $J_t$ be the total flow through $\Lambda$ up to time $t$. This current can be seen as the total number of particles entering at the left boundary (where a particle leaving is counted with a negative sign); up to accumulation of particles this is the same as counting particles exiting at the right boundary.
Then the expected flow starting from stationarity is
\[
\mathbb{E}_{\mu_\alpha} [J_t] = K  (\alpha_r-\alpha_l) \, t,
\]
for $K=\frac{L}{L+1}$. This shows \cite[Equation (13)]{ERSS24}.
\end{corollary}

\section{Critical exponents in $d=1$} \label{sec:d1}

The one-dimensional case is unlike the higher-dimensional one, because in $ d=1$ the deterministic threshold $\rho=1/2$ matches with the critical density $\rho_c$: because of the much simpler structure of frozen configurations in one dimension (any alternate configuration is frozen), there is no potential barrier to overcome to reach a frozen state at densities $\rho\leq 1/2$. Furthermore, in one dimension, the grand canonical states are explicit, so that all quantities above have explicit expressions as functions of $\rho$. This section is therefore dedicated to fully describing the one-dimensional case.

\medskip

First, we describe the infinite volume grand-canonical states in one dimension. We say that a configuration is ergodic on $\Lambda\subset \Z$  if for any edge $(i,i+1)\subset \Lambda$, $\eta_i+\eta_{i+1}\geq 1$. For $\rho\geq 1/2$ the stationary measures are given as the translation-invariant measures $\pi_\rho$ on $\{0,1\}^\Z$ whose marginal on $B_\ell:=\llbracket 1,\ell \rrbracket$ is given for any $\sigma\in \{0,1\}^{B_\ell}$ by 
\[\pi_\rho(\eta_{|B_\ell}=\sigma)=(1-\rho)\rho_a^{2p-\ell+1-\sigma_1-\sigma_\ell} (1-\rho_a)^{\ell-1-p}{\bf1}_{\{\sigma \mbox{ is ergodic on } B_\ell\}},\]
where $p=\sum_{i=1}^\ell \sigma_i$ is the number of particles in $\sigma$, and the active density is related to the density by
\[\rho_a(\rho)=\frac{2\rho-1}{\rho}.\]
Equivalently to the previous formula, $\pi_\rho$ can be be seen as the distribution of a Markov chain on $\{0,1\}$, with initial $\eta_1\sim \mathrm{Ber}(\rho)$, and with transitions
\[\pi_\rho(\eta_{i+1}=1\mid \eta_i=1)=\rho_a, \qquad \mbox{ and } \qquad \pi_\rho(\eta_{i+1}=1\mid \eta_i=0)=1.\]

\medskip

After straightforward computations, these explicit constructions of the grand-canonical states yield the following explicit expressions for the macroscopic observables described above, namely;
\[\begin{array}{cc}
\medskip
\mathfrak{a}(\rho)=2\rho_a(1-\rho)<\rho_a &\qquad \mbox{(Activity)}\\
\medskip
D(\rho)=\rho_a'=\frac{1}{\rho^2} &\qquad \mbox{(Diffusion coefficient, \cite{EZ2023})}
\\
\medskip
\varphi_\rho(0,i)=\rho(1-\rho)(\rho_a-1)^i &\qquad \mbox{(Two-points correlation function,  \cite[Appendix A]{EZ2023})}\\
\medskip
\chi(\rho)=\rho(1-\rho)(2\rho-1)& \qquad \mbox{(Compressibility, \cite{EZ2023})}\\
\medskip
\sigma(\rho)=\frac{(1-\rho)(2\rho-1)}{\rho}& \qquad \mbox{(Conductivity, \cite{EZ2023})}\\
\medskip
\xi_\times(\rho)=-\log(1-\rho_a)^{-1}\underset{\rho\to 1/2}{\sim} 1/\rho_a& \qquad \mbox{(Two-points correlation length)}\\
\medskip
\xi_\perp(\rho)= 1/\rho_a& \qquad \mbox{(Geometric correlation length)}
\end{array}\]
In one dimension, under $\pi_\rho$, the  probability that a jump can occur over an edge in a fixed direction is $\rho_a(\rho)(1-\rho)$, summing over the two possible directions yields $\mathfrak{a}(\rho)$. 
Furthermore, the two-points correlation length $\xi_\times(\rho)$ is immediately obtained from the two-points correlation function $\varphi_\rho $ as its inverse exponential decay rate.
The geometric correlation length is the only one that is not explicitly computed through the grand canonical states.  However it can be straightforwardly obtained because of the lack of fluctuations of the critical state for the one-dimensional state: since the latter must be an alternate configuration, the off-criticality cannot be hidden, and the standard variation of the number of particles in a critical box vanishes, so that $\zeta=0$. In particular, over length scales larger than $1/\rho_a$, an active particle has typically been found in the grand-canonical state, and the off-criticality is evidenced. Since close to criticality $\rho_a\sim \rho-1/2$, this yields the critical exponents 
\[b=\beta=\gamma=\nu_\times=\nu_\perp=1.\]
Finally, since the diffusion coefficient converges to $1/4$ as $\rho\to 1/2$, we also have $\alpha=0$.

\bibliographystyle{plain}
\bibliography{bibliography.bib}

\begin{thebibliography}{10}

\bibitem{AGLS10}
A.~Ayyer, S.~Goldstein, J.~L. Lebowitz, and E.~R. Speer.
\newblock Stationary states of the one-dimensional facilitated asymmetric
  exclusion process.
\newblock {\em arXiv:2010.07257}.

\bibitem{BBCS}
J.~Baik, G.~Barraquand, I.~Corwin, and T.~Suidan.
\newblock {\em Facilitated Exclusion Process: The Abel Symposium, Rosendal,
  Norway, August 2016}, pages 1--35.
\newblock 01 2018.

\bibitem{BESS}
O.~Blondel, C.~Erignoux, M.~Sasada, and M.~Simon.
\newblock {Hydrodynamic limit for a facilitated exclusion process}.
\newblock {\em Ann. Inst. H. Poincar\'{e} Probab. Statist.}, 56(1):667--714,
  2020.

\bibitem{BES}
O.~Blondel, C.~Erignoux, and M.~Simon.
\newblock Stefan problem for a nonergodic facilitated exclusion process.
\newblock {\em Probability and Mathematical Physics}, 2(1):127--178, 2021.

\bibitem{DES24}
H.~Da~Cunha, C.~Erignoux, and M.~Simon.
\newblock Hydrodynamic limit for a boundary-driven facilitated exclusion
  process.
\newblock {\em arXiv: 2401.16535}, 2024.

\bibitem{O}
M.~J. de~Oliveira.
\newblock Conserved lattice gas model with infinitely many absorbing states in
  one dimension.
\newblock {\em Phys. Rev. E}, 71:016112, Jan 2005.

\bibitem{ERSS24}
C.~Erignoux, A.~Roget, A.~Shapira, and M.~Simon.
\newblock Hydrodynamic behavior near dynamical criticality of a facilitated
  conservative lattice gas.
\newblock {\em Phys. Rev. E}, 110:L032101, Sep 2024.

\bibitem{EZ2023}
C.~Erignoux and L.~Zhao.
\newblock Stationary fluctuations for the facilitated exclusion process.
\newblock {\em ArXiv preprint: 2305.13853}, 2023.

\bibitem{GKR}
A.~Gabel, P.~L. Krapivsky, and S.~Redner.
\newblock Facilitated asymmetric exclusion.
\newblock {\em Phys. Rev. Lett.}, 105(21):210603, 4, 2010.

\bibitem{Gold21}
S.~Goldstein, J.~L. Lebowitz, and E.~R. Speer.
\newblock The discrete-time facilitated totally asymmetric simple exclusion
  process.
\newblock {\em Pure Appl. Funct. Anal.}, 6.

\bibitem{Gold2019}
S.~Goldstein, J.~L. Lebowitz, and E.~R. Speer.
\newblock Exact solution of the facilitated totally asymmetric simple exclusion
  process.
\newblock {\em Journal of Statistical Mechanics: Theory and Experiment},
  2019(12):123202, dec 2019.

\bibitem{GLS25}
S.~Goldstein, J.~L. Lebowitz, and E.~R. Speer.
\newblock Approach to hyperuniformity in the one-dimensional facilitated
  exclusion process.
\newblock {\em arXiv: 2407.02652}, 2024.

\bibitem{GLS2024}
S.~Goldstein, J.~L. Lebowitz, and E.~R. Speer.
\newblock Approach to hyperuniformity of steady states of facilitated exchange
  processes.
\newblock {\em J Phys Condens Matter}, 36, May 2024.

\bibitem{HL}
D.~Hexner and D.~Levine.
\newblock Hyperuniformity of critical absorbing states.
\newblock {\em Phys. Rev. Lett.}, 114:110602, Mar 2015.

\bibitem{KL}
C.~Kipnis and C.~Landim.
\newblock {\em Scaling limits of interacting particle systems}, volume 320 of
  {\em Grundlehren der Mathematischen Wissenschaften [Fundamental Principles of
  Mathematical Sciences]}.
\newblock Springer-Verlag, Berlin, 1999.

\bibitem{Lubeck}
S.~L\"ubeck.
\newblock Scaling behavior of the absorbing phase transition in a conserved
  lattice gas around the upper critical dimension.
\newblock {\em Phys. Rev. E}, 64:016123, Jun 2001.

\bibitem{RPV}
M.~Rossi, R.~Pastor-Satorras, and A.~Vespignani.
\newblock Universality class of absorbing phase transitions with a conserved
  field.
\newblock {\em Phys. Rev. Lett.}, 85:1803--1806, Aug 2000.

\bibitem{Spohn}
H.~Spohn.
\newblock {\em Large Scale Dynamics of Interacting Particles}.
\newblock Springer-Verlag, 1991.

\bibitem{ZC2019}
L.~Zhao and D.~Chen.
\newblock The invariant measures and the limiting behaviors of the facilitated
  tasep.
\newblock {\em Statistics and Probability Letters}, 154:108557, 2019.

\end{thebibliography}

\end{document}